\documentclass[letterpaper,10pt,twocolumn,twoside,conference]{IEEEtran}
\IEEEoverridecommandlockouts                         
\usepackage{hyperref}
\usepackage{graphicx}		% For \includegraphics
\usepackage[format=plain,font=footnotesize,labelfont=bf,labelsep=period]{caption}
\usepackage{sidecap} % For sidecaptions for figures
\usepackage[export]{adjustbox}
\usepackage{cite}
\usepackage{amsmath}
\usepackage{amssymb}
\usepackage{mathtools}
\usepackage{algorithm}
\usepackage{algpseudocode}

\usepackage{pdfsync}
\usepackage[normalem]{ulem}
\usepackage{paralist}	% For enumerations with better titled numbers
\usepackage[space]{grffile} % Space in filenames

\usepackage{color}
\usepackage{wrapfig}
\usepackage[format=plain,font=footnotesize,labelfont=bf,labelsep=period]{caption}
\usepackage{sidecap} % For sidecaptions for figures
\usepackage[export]{adjustbox}
\usepackage{dblfloatfix}

\usepackage{amsthm}

\newtheorem{theorem}{Theorem}

\theoremstyle{definition}
\newtheorem{definition}{Definition}
\theoremstyle{remark}
\newtheorem{remark}{Remark}
\theoremstyle{definition}

\theoremstyle{definition}

\newtheorem{example}{Example}

\newcommand{\R}{\mathbb{R}}

\newcommand{\C}{\mathcal{C}}

\definecolor{darkblue}{RGB}{0,0,102}
\definecolor{lightblue}{RGB}{77,77,148}

\definecolor{gold}{RGB}{234, 170, 0}
\definecolor{metallic_gold}{RGB}{139, 111, 78}

\newcommand{\defeq}{\triangleq}
\renewcommand{\cal}[1]{\mathcal{ #1 }}
\newcommand{\mb}[1]{\mathbf{ #1 }}
\newcommand{\bs}[1]{\boldsymbol{ #1 }}

\DeclareMathOperator*{\argmin}{argmin}

\newcommand{\derp}[2]{\frac{\partial #1 }{\partial #2 }}

\newcommand{\dVdx}{\frac{\partial V_{0}}{\partial \mb{x}}(\mb{x})}
\newcommand{\dhdx}{\frac{\partial h_{0}}{\partial \mb{x}}(\mb{x})}
\newcommand{\kd}{\mb{k}_{\rm 0}(\mb{x})}
\newcommand{\kdh}{\mb{k}_{\rm 0}(\mb{x})}
\newcommand{\kdV}{\mb{k}_{\rm 0}(\mb{x})}

\newcommand{\dkVdx}{\frac{\partial \mb{k}_{\rm 0}}{\partial \mb{x}}(\mb{x})}
\newcommand{\dkhdx}{\frac{\partial \mb{k}_{\rm 0}}{\partial \mb{x}}(\mb{x})}
\newcommand{\dotx}{\mb{f}_0(\mb{x}) + \mb{g}_0(\mb{x}) \bs{\xi}}
\newcommand{\dotxi}{\mb{f}_1(\mb{x},\bs{\xi}) + \mb{g}_1(\mb{x},\bs{\xi}) \mb{u}}

\newcommand{\setdefb}[2]{\big\{#1 \; | \; #2\big\}}

\newcommand{\Uc}{\mathcal{U}}

\newcommand{\map}[3]{#1:#2 \rightarrow #3}

\title{\LARGE \textbf{Safe Backstepping with Control Barrier Functions}}
\author{Andrew J. Taylor, Pio Ong, Tamas G. Molnar, and Aaron D. Ames
\thanks{This research is supported in part by Ford, the National Science Foundation (CPS Award \#1932091,  CMMI Award \#1923239), Raytheon Technologies, Aerovironment and Dow (\#227027AT).}
\thanks{
The authors are with the Department of Computing and Mathematical Sciences and the Department of Mechanical and Civil Engineering, California Institute of Technology, Pasadena, CA 91125, USA, {\tt\small \{ajtaylor, pioong, tmolnar, ames\}@caltech.edu}.
}
}

\begin{document}

\maketitle

\begin{abstract}%
Complex control systems are often described in a layered fashion, represented as \emph{higher-order systems} where the inputs appear after a chain of integrators.
While Control Barrier Functions (CBFs) have proven to be powerful tools for safety-critical controller design of nonlinear systems, their application to higher-order systems adds complexity to the controller synthesis process---it necessitates dynamically extending the CBF to include higher order terms, which consequently modifies the safe set in complex ways.   
% While Control Barrier Functions (CBFs) have proven to be powerful tools for safety-critical controller design of nonlinear systems, their application to higher-order systems adds complexity to the controller synthesis process through the dynamic extension of the CBF to include higher order terms, and as a consequence modify the safe set.   
% Control Barrier Functions (CBFs) have been demonstrated
% to be powerful tools for safety-critical controller design
% for nonlinear systems. Existing design paradigms for higher-order systems require verifying CBF conditions with the full-system dynamics, or impose structural requirements on safety constraints by coupling them with stability of a system. 
We propose an alternative approach for addressing safety of higher-order systems through \textit{Control Barrier Function Backstepping}. Drawing inspiration from the method of Lyapunov backstepping, we provide a constructive framework for synthesizing safety-critical controllers and CBFs for higher-order systems from a top-level dynamics safety specification and controller design. Furthermore, we integrate the proposed method with Lyapunov backstepping, allowing the tasks of stability and safety to be expressed individually but achieved jointly. We demonstrate the efficacy of this approach in simulation.

% Control Barrier Functions (CBFs) have been demonstrated
% to be powerful tools for safety-critical controller design
% for nonlinear systems. Existing design paradigms for higher-order systems require verifying CBF conditions with the full-system dynamics, or impose structural requirements on safety constraints by coupling them with stability of a system. We propose an alternative approach for addressing safety of higher-order systems through \textit{Control Barrier Function Backstepping}. Drawing inspiration from the method of Lyapunov backstepping, we provide a constructive framework for synthesizing safety-critical controllers and CBFs for higher-order systems from a top-level dynamics safety specification and controller design. Furthermore, we integrate the proposed method with Lyapunov backstepping, allowing the tasks of stability and safety to be expressed individually but achieved jointly. We demonstrate the efficacy of this approach in simulation.
\end{abstract}

\section{Introduction}

% \par
% % 1- Deal with complex systems safety-critical fluff
% Safety is of the utmost importance in modern control systems. As the complexity of systems continues to grow, it is often desirable to approach the control design process with a simplified top-level model that guides design for subsystems addressing the full system dynamics. Thus it is essential to develop safety-critical control synthesis techniques that allow us to systematically address design at multiple levels with varying degrees of model complexity.

\par
% 2- CBFs have provided a powerful tool for getting safety, extended to higher order systems
Safety is becoming an ever more prevalent design consideration in modern control systems as these systems are deployed in real-world environments. % working around and with humans. 
Control Barrier Functions (CBFs) have become a popular tool for constructively synthesizing controllers that endow nonlinear systems with rigorous guarantees of safety \cite{ames2014control, ames2019control}. Originally posed such that the input of the system directly impacted the time derivative of the CBF, recent work has sought to extend this to higher-order nonlinear systems in which multiple time derivatives are required for the input to influence the evolution of the CBF \cite{nguyen2016exponential, xiao2019control, xiao2021high,breeden2021high}. While these works allow for the safety-critical control of higher-order systems, they require verifying the feasibility of CBF conditions using the full system dynamics and change the safe set in complex ways. %, which may be difficult for complex systems.
Alternatively, the work in \cite{molnar2021model} has explored designing CBFs for a top-level model, and using a tracking controller that addresses the full system dynamics.

\par
% 3- Backstepping is a well established technique for nonlinear systems, used for Lyapunov and barrier Lyapunov, not directly CBFs
As the complexity of systems increase, it is often desirable to approach the control design process with a simplified top-level model that guides design for subsystems addressing the full system dynamics.
Backstepping is a well established design technique for addressing the robust stabilization of layered systems of this form, i.e., nonlinear systems with \emph{higher-order dynamics} \cite{freeman1993backstepping, sepulchre2012constructive}. It considers design for the top-level model and recursively designs a controller using the full system dynamics, also allowing it to address the challenge of \textit{mixed-relative degree}, where inputs enter the system dynamics at different levels. Using backstepping to stabilize systems while meeting state constraints has been studied through lens of non-overshooting control \cite{krstic2006nonovershooting}, and has recently been related to CBFs \cite{abel2022prescribed, koga2021safe}. These works achieve safe behavior using a structured controller that yields a linear dynamic relationship between sequential states in a cascade, such that a system does not overshoot a setpoint as it stabilizes. Other work has used backstepping in the context of Lyapunov-Barrier functions to ensure state constraints are met \cite{ngo2005integrator, xia2019backstepping, fu2020barrier}. These approaches couple ensuring safety with ensuring stability, which may impose strict structural requirements on safety constraints. To the best of our knowledge, decoupling stability and safety and exploring backstepping purely with safety constraints expressed through CBFs has not been considered.

\par 
% 4- Backstepping typically needs smooth controllers as they get differentiated, but CBFs are often implemented with non-smooth optimization problems.

A core challenge in combining CBF-based methods with backstepping is finding smooth controllers that ensure safety as backstepping requires the differentiation of controllers appearing higher in the integrator chain.   
From the conception of CBFs, they have typically been used as constraints in optimization-based controllers---either paired with CLFs \cite{ames2014control}, or filtering a desired stabilizing controller \cite{ames2019control}---and therefore are inherently non-smooth. 
Additionally, typically one wishes to design controllers that are not only safe, but are also stabilizing, precluding smooth CBF controller instantiations, e.g., using Sontag's Universal formula \cite{sontag1989universal}.
 While it may be possible to address these non-smooth challenges \cite{tanner2003backstepping,glotfelter2017nonsmooth}, we will consider the approach in \cite{ong2019universal} for synthesizing smooth controllers meeting both CLF and CBF constraints.

% The goal of this paper is to to unify backstepping with CBFs, thereby enabling safe controller design at multiple levels with varying degrees of model complexity.  One of the core challenges in combining CBF-based methods with backstepping lies in finding smooth controllers that ensure safety as backstepping requires the differentiation of controllers appearing higher in the integrator chain.   
% From the conception of CBFs, they have been typically been used as constraints in optimization-based controllers---either paired with CLFs \cite{ames2014control}, or filtering a desired stabilizing controller \cite{ames2019control}---and therefore are inherently non-smooth. 
% Additionally, typically one wishes to design controllers that are not only safe, but are also stabilizing, precluding smooth CBF controller instantiations, e.g., using Sontag's Universal formula \cite{sontag1989universal}.
%  While it may be possible to address these non-smooth challenges \cite{tanner2003backstepping,glotfelter2017nonsmooth}, we will consider the approach in \cite{ong2019universal} for synthesizing smooth controllers meeting both CLF and CBF constraints.

%The resulting controllers are often not continuously differentiable, preventing them from easily being incorporated into a backstepping approach. 

\par
% 5- Contributions

The goal of this paper is to to unify backstepping with CBFs, thereby enabling safe controller design at multiple levels with varying degrees of model complexity.  
To this end, after a review of CBFs and Lyapunov backstepping, we begin in Section \ref{sec:barbackstep} by formulating a nonlinear controller that ensures safety of a system with a single cascade via Barrier Functions and backstepping.
A consequence of this result is that we may constructively synthesize a CBF for the full cascaded system using a CBF and smooth controller designed only considering the top-level of the system, which is often easier than directly finding a CBF for the full-order system.  Additionally, in Section \ref{sec:multibarbackstep}, we demonstrate that this approach can be generalized to the multiple-cascade setting, and address the challenge of mixed relative-degree systems.
The main result of this paper, presented in Section \ref{sec:design}, is the unification of Lyapunov and Barrier backstepping, wherein we show that by designing a controller that renders the top-level dynamics both stable and safe, we may use backstepping to achieve stability and safety of the full cascaded system. 
Importantly, using the  techniques in \cite{ong2019universal}, we are able to design a smooth top-level controller amenable to backstepping. 
%To design a smooth top-level controller amenable to backstepping, we use the techniques in \cite{ong2019universal}.
These results are demonstrated in simulation in Section \ref{sec:sim} on multiple examples in the context of obstacle avoidance.

\clearpage

\section{Background}
\label{sec:background}
In this section we revisit Barrier Functions, Control Barrier Functions and Lyapunov backstepping as a precursor to introducing Control Barrier Function backstepping.

%\subsection{System Dynamics}
Consider a nonlinear control-affine system: 
\begin{equation}
    \label{eq:openloop1}
    \dot{\mb{x}} = \mb{f}(\mb{x})+\mb{g}(\mb{x})\mb{u},
\end{equation}
with state $\mb{x}\in\R^n$, input $\mb{u}\in\R^m$, and functions $\mb{f}:\R^n\to\R^n$ and $\mb{g}:\R^n\to\R^{n\times m}$ assumed to be locally Lipschitz continuous on $\R^n$. A locally Lipschitz continuous controller $\mb{k}:\R^n\to\R^m$ yields the closed loop system:
\begin{equation}
    \label{eq:closedloop1}
    \dot{\mb{x}} = \mb{f}(\mb{x})+\mb{g}(\mb{x})\mb{k}(\mb{x}).
\end{equation}
As the functions $\mb{f}$, $\mb{g}$, and $\mb{k}$ are locally Lipschitz continuous, for any initial condition $\mb{x}_0 \in \R^n$, there exists a maximal time interval $I(\mb{x}_0)=[0,t_{\rm max}(\mb{x}_0))$ and a unique continuously differentiable solution $\bs{\varphi}:I(\mb{x}_0)\to\R^n$ satisfying:
\begin{align}
\label{eqn:soldiff}
    \dot{\bs{\varphi}}(t) &= \mb{f}(\bs{\varphi}(t)) + \mb{g}(\bs{\varphi}(t))\mb{k}(\bs{\varphi}(t)), \\ \label{eqn:solic}
    \bs{\varphi}(0) &= \mb{x}_0,
\end{align}
for all $t\in I(\mb{x}_0)$ \cite{perko2013differential}.

\subsection{Control Barrier Functions}
We define the notion of safety in this context as forward invariance of a set in the state space. Specifically, suppose there exists a set $\C\subset \R^n$ defined as the 0-superlevel set of a continuously differentiable function $h:\R^n \to \R$:
\begin{align}
    \label{eq:C1} \C &= \left\{\mb{x} \in \R^n ~|~ h(\mb{x}) \geq 0\right\}.
\end{align}
The set $\C$ is said to be \emph{forward invariant} if for any initial condition $\mb{x}_0 \in \C$, we have $\bs{\varphi}(t)\in\C$ for all $t\in I(\mb{x}_0)$. In this case, we call the system \eqref{eq:closedloop1}  \emph{safe} with respect to the set $\C$, and refer to $\C$ as the \emph{safe set}.

Before defining Barrier Functions and Control Barrier Functions, we recall the following definitions. A continuous function $\alpha:[0,\infty)\to[0,\infty)$ is said to be \emph{class $\cal{K}_\infty$} ($\alpha\in\cal{K}_{\infty}$) if $\alpha$ is strictly monotonically increasing with $\alpha(0)=0$ and $\lim_{r\to\infty}\alpha(r)=\infty$, and a continuous function $\alpha:\R\to\R$ is said to be \emph{extended class $\cal{K}_\infty$} ($\alpha\in\cal{K}_{\infty}^{\rm e}$) if it belongs to $\cal{K}_\infty$ and $\lim_{r\to-\infty}\alpha(r)=-\infty$. We now define Barrier Functions:
\begin{definition}[\textit{Barrier Function (BF)} \cite{ames2017control}]
Let $\C\subset\R^n$ be the 0-superlevel set of a continuously differentiable function $h:\R^n\to\R$ with $\derp{h}{\mb{x}}(\mb{x}) \neq \mb{0}$ when $h(\mb{x})=0$. The function $h$ is a \emph{Barrier Function} (BF) for \eqref{eq:closedloop1} on $\C$ if there exists $\alpha\in\mathcal{K}_{\infty}^{\rm e}$ such that for all $\mb{x}\in\R^n$:
\begin{equation}
\label{eq:bf} \underbrace{\derp{h}{\mb{x}}(\mb{x})\mb{f}(\mb{x})}_{L_\mb{f}h(\mb{x})}+\underbrace{\derp{h}{\mb{x}}(\mb{x})\mb{g}(\mb{x})}_{L_\mb{g}h(\mb{x})}\mb{k}(\mb{x})\geq -\alpha(h(\mb{x})).
\end{equation}
\end{definition}
We have the following result establishing the safety of a set $\C$ for the closed-loop system \eqref{eq:closedloop1} through Barrier Functions:
\begin{theorem}[\cite{ames2017control, konda2020characterizing}]
\label{thm:bfsafety}
Let $\C\subset\R^n$ be the $0$-superlevel set of a continuously differentiable function $h:\R^n\to\R$ with $\derp{h}{\mb{x}}(\mb{x})\neq\mb{0}$ when $h(\mb{x})=0$. If $h$ is a BF for \eqref{eq:closedloop1} on $\C$, then the system \eqref{eq:closedloop1} is safe with respect to the set $\C$.
\end{theorem}

Control Barrier Functions provide a tool for synthesizing controllers that enforce the safety of $\C$:

\begin{definition}[\textit{Control Barrier Function (CBF)} \cite{ames2017control}]
Let $\C\subset\R^n$ be the 0-superlevel set of a continuously differentiable function $h:\R^n\to\R$ with $\derp{h}{\mb{x}}(\mb{x}) \neq \mb{0}$ when $h(\mb{x})=0$. The function $h$ is a \emph{Control Barrier Function} (CBF) for \eqref{eq:openloop1} on $\C$ if there exists $\alpha\in\mathcal{K}_{\infty}^{\rm e}$ such that for all $\mb{x}\in\R^n$:
\begin{equation}
\label{eq:cbf}
     \sup_{\mb{u}\in\R^m} \dot{h}(\mb{x},\mb{u}) \triangleq\sup_{\mb{u}\in\R^m} {L_\mb{f}h(\mb{x})}+{L_\mb{g}h(\mb{x})}\mb{u}>-\alpha(h(\mb{x})).
\end{equation}
\end{definition}

\noindent Given a CBF $h$ for \eqref{eq:openloop1} and a corresponding ${\alpha\in\cal{K}_{\infty}^{\rm e}}$, we define the point-wise set of control values:
\begin{equation}
\label{eqn:Kcbf}
    K_{\textrm{CBF}}(\mb{x}) = \left\{\mb{u}\in\R^m ~\left|~ \dot{h}(\mb{x},\mb{u})\geq-\alpha(h(\mb{x})) \right. \right\}.
\end{equation}
This yields the following result:

\begin{theorem}[\hspace{-0.1 mm}\cite{ames2017control}] \label{thm:safety}
Let ${\C\subset\R^n}$ be the 0-superlevel set of a continuously differentiable function ${h:\R^n\to\R}$ with ${\derp{h}{\mb{x}}(\mb{x}) \neq \mb{0}}$ when ${h(\mb{x})=0}$. If $h$ is a CBF for \eqref{eq:openloop1} on $\C$, then the set $K_{\rm CBF}(\mb{x})$ is non-empty for all $\mb{x}\in\R^n$, and for any locally Lipschitz continuous controller $\mb{k}$ with ${\mb{k}(\mb{x}) \in K_{\rm CBF }(\mb{x})}$ for all ${\mb{x}\in\R^n}$, the function $h$ is a BF for \eqref{eq:closedloop1} on $\C$.
\end{theorem}

\begin{remark}
The strict inequality in \eqref{eq:cbf} serves two purposes. First, it ensures the set \eqref{eqn:Kcbf} is non-empty (as with a non-strict inequality in \eqref{eq:bf}, the supremum may hold with equality, but there may be no input such that the supremum is attained). Second, strictness enables proving optimization-based controllers using CBFs are locally Lipschitz continuous \cite{jankovic2018robust}.
\end{remark}

\subsection{Lyapunov Backstepping}
Consider now a nonlinear control-affine system of the form:
\begin{align}
    \dot{\mb{x}} &= \mb{f}_0(\mb{x}) + \mb{g}_0(\mb{x})\bs{\xi}, \label{eqn:xdot} \\
    \dot{\bs{\xi}} &= \mb{f}_1(\mb{x},\bs{\xi})+\mb{g}_1(\mb{x},\bs{\xi})\mb{u}, \label{eqn:xidot}
\end{align}
with $\mb{x}\in\R^n$, $\bs{\xi}\in\R^p$, and $\mb{u}\in\R^m$, and functions $\mb{f}_0:\R^n\to\R^n$, $\mb{g}_0:\R^n\to\R^{n\times p}$, $\mb{f}_1:\R^n\times\R^p\to\R^p$, and $\mb{g}_1:\R^n\times\R^p\to\R^{p\times m}$ assumed to be locally Lipschitz continuous on their respective domains. This system is referred to as being in \textit{strict-feedback form}. We further assume that $\mb{f}_0(\mb{0})=\mb{f}_1(\mb{0},\mb{0})=\mb{0}$ and $\mb{g}_1$ is pseudo-invertible on $\R^n\times\R^p$. As before, given a locally Lipschitz continuous feedback controller $\mb{k}:\R^n\times\R^p\to\R^m$ yielding the closed-loop system:
\begin{align}
    \dot{\mb{x}} &= \mb{f}_0(\mb{x}) + \mb{g}_0(\mb{x})\bs{\xi}, \label{eqn:xdotcl} \\
    \dot{\bs{\xi}} &= \mb{f}_1(\mb{x},\bs{\xi})+\mb{g}_1(\mb{x},\bs{\xi})\mb{k}(\mb{x},\bs{\xi}), \label{eqn:xidotcl}
\end{align}
for any initial condition $(\mb{x}_0,\bs{\xi}_0)\in\R^n\times\R^p$ there exists a maximum time interval $I((\mb{x}_0,\bs{\xi}_0))\subseteq\R_{\geq 0}$ and a unique solution denoted by $\bs{\varphi}=(\bs{\varphi}_{\mb{x}},\bs{\varphi}_{\bs{\xi}})$ satisfying \eqref{eqn:soldiff}-\eqref{eqn:solic}
$\forall t\in I((\mb{x}_0,\bs{\xi}_0))$.

Suppose there exist a function $V_0:\R^n\to\R_{\geq 0}$ and a function $\kdV:\R^n\to\R^p$, both twice-continuously differentiable on $\R^n$, and $\gamma_1,\gamma_2,\gamma_3\in\mathcal{K}_{\infty}$ such that $\mb{k}_0(\mb{0})=\mb{0}$ and:
\begin{align}
    \gamma_1(\Vert\mb{x}\Vert_2) \leq V_0(\mb{x}) &\leq \gamma_2(\Vert \mb{x} \Vert_2), \label{eqn:V0bds} \\
    \derp{V_0}{\mb{x}}(\mb{x})(\mb{f}_0(\mb{x})+\mb{g}_0(\mb{x})\mb{k}_0(\mb{x})) &\leq -\gamma_3(\Vert\mb{x}\Vert_2), \label{eq:stability_high_level}
\end{align}
for all $\mb{x}\in\R^n$. The function $\kdV$ reflects a stabilizing controller that we would implement for the system \eqref{eqn:xdot} if we could directly control $\bs{\xi}$. As we may only directly control $\mb{u}$, we must \emph{backstep} through the state $\bs{\xi}$ to access $\mb{u}$. More precisely, consider a function $V:\R^n\times\R^p\to\R_{\geq 0}$ defined as:
\begin{equation}
\label{eqn:Vcomp}
    V(\mb{x},\bs{\xi}) = V_0(\mb{x}) + \frac{1}{2 \mu} (\bs{\xi} - \kdV)^\top (\bs{\xi} - \kdV),
\end{equation}
where $\mu\in\R_{> 0}$. We note there exists $\gamma_1',\gamma_2'\in\mathcal{K}_{\infty}$ such that:
\begin{align}
     \gamma_1(\Vert\mb{x}\Vert_2)+\gamma_1'(\Vert\bs{\xi} - \kdV\Vert_2) \leq V(\mb{x},\bs{\xi}), \\ V(\mb{x},\bs{\xi}) \leq \gamma_2(\Vert\mb{x}\Vert_2)+\gamma_2'(\Vert\bs{\xi} - \kdV\Vert_2),
\end{align}
for all $\mb{x}\in\R^n$ and $\bs{\xi}\in\R^p$. The time derivative of $V$ is:
\begin{align}
    \dot{V}(\mb{x},\bs{\xi},\mb{u}) = ~&
    \dVdx \big( \dotx \big) \label{eqn:Vcomp_dt}\\
    &+ \frac{1}{\mu} (\bs{\xi} - \kdV)^\top \bigg( \dotxi \nonumber \\
    &\qquad\qquad - \dkVdx \big( \dotx \big) \bigg). \nonumber
\end{align}
Using a locally Lipschitz continuous feedback controller $\mb{k}:\R^n\times\R^p\to\R^m$ defined as:
\begin{align}
    \hspace{-0.4 mm}\mb{k}(\mb{x},\bs{\xi}) = &~
    \mb{g}_1(\mb{x},\bs{\xi})^\dagger
    \bigg( - \mb{f}_1(\mb{x},\bs{\xi})
    + \dkVdx \big( \dotx \big) \nonumber\\
    & - \mu \left(\dVdx \mb{g}_0(\mb{x}) \right)^\top
    - \frac{\lambda}{2} (\bs{\xi} - \kdV) \bigg),
\end{align}
with $\lambda\in\R_{\geq 0}$ yields:
\begin{align} \label{eqn:dotVgammas}
    \dot{V}(\mb{x},\bs{\xi},\mb{k}(\mb{x},\bs{\xi})) = &~ \dVdx(\mb{f}(\mb{x})+\mb{g}(\mb{x})\kdV) \\ &-\frac{\lambda}{2 \mu}
    (\bs{\xi} - \kdV)^\top (\bs{\xi} - \kdV), \nonumber \\
    \leq &~-\gamma_3(\Vert\mb{x}\Vert_2)-\gamma_3'(\Vert
    \bs{\xi} - \kdV\Vert_2), 
\end{align}
for $\gamma_3'\in\mathcal{K}_\infty$ defined as $\gamma_3'(s) \triangleq \lambda/(2\mu)s^2$. Hence $V$ is a Lyapunov function for \eqref{eqn:xdotcl}-\eqref{eqn:xidotcl}, such that $I((\mb{x}_0,\bs{\xi}_0))=[0,\infty)$ for all $(\mb{x}_0,\bs{\xi}_0)\in\R^n\times\R^p$, and $\bs{\varphi}_{\mb{x}}(t)\to\mb{0}$ and $\bs{\varphi}_{\bs{\xi}}(t)-\mb{k}_0(\bs{\varphi}_{\mb{x}}(t))\to\mb{0}$ as $t\to\infty$. Furthermore, we have:
\begin{equation}
    \inf_{\mb{u}\in\R^m} \dot{V}(\mb{x},\bs{\xi},\mb{u}) < -c\gamma_3(\Vert\mb{x}\Vert_2)-c\gamma_3'(\Vert
    \bs{\xi} - \kdV\Vert_2),
\end{equation}
for all $\mb{x}\neq\mb{0}$ and $\bs{\xi}\neq\kdV$, where $c\in(0,1)$, such that $V$ is a Control Lyapunov Function (CLF) \cite{jankovic2018robust}. This enables a convex optimization-based controller defined as follows:
\begin{align}
    \mb{k}(\mb{x},\bs{\xi}) & =  \argmin_{\mb{u}\in\R^m} \frac{1}{2}\Vert \mb{u} \Vert_2^2 \\ & \textrm{s.t.}~ \dot{V}(\mb{x},\bs{\xi},\mb{u}) \leq -c\gamma_3(\Vert\mb{x}\Vert_2)-c\gamma_3'(\Vert
    \bs{\xi} - \kdV\Vert_2), \nonumber
\end{align}
that stabilizes \eqref{eqn:xdotcl}-\eqref{eqn:xidotcl} and is locally Lipschitz continuous on $(\R^n\times\R^p)\setminus\{\mb{0}\}$ if $\gamma_3$ is locally Lipschitz continuous \cite{jankovic2018robust}.

\section{Control Barrier Function Backstepping}
\label{sec:barbackstep}
In this section we explore how Control Barrier Functions can be used to achieve safety for the cascaded system in \eqref{eqn:xdot}-\eqref{eqn:xidot} when one must backstep through the state $\bs{\xi}$.

Suppose there exists a set $\C_0\subset\R^n$ defined as the 0-superlevel set of a twice-continuously differentiable function $h_0:\R^n\to\R$:
\begin{equation}
    \C_0 =\{\mb{x}\in\R^n ~|~ h_0(\mb{x})\geq 0\},
    \label{eq:safeset}
\end{equation}
that we wish to keep safe. We further assume that $\derp{h_0}{\mb{x}}(\mb{x})\neq\mb{0}$ when $h_0(\mb{x})=0$. As the input $\mb{u}$ does not show up in the time derivative of $h_0$, we may not directly apply the Control Barrier Function methodology established in Section \ref{sec:background}. Instead, motivated by the Lyapunov setting, we take a backstepping approach using CBFs. In particular, suppose there exists a twice-continuously differentiable function $\kdh:\R^n\to\R^p$ and a function $\alpha_0\in\mathcal{K}_{\infty}^e$ such that:
\begin{equation}
    \derp{h_0}{\mb{x}}(\mb{x})\left(\mb{f}_0(\mb{x})+\mb{g}_0(\mb{x})\kdh \right)\geq -\alpha_0(h_0(\mb{x})).
    \label{eq:safety_high_level}
\end{equation}
As before, $\kdh$ reflects a controller that renders $\C_0$ safe that we would implement for the system \eqref{eqn:xdot} if we could directly control $\bs{\xi}$. Let us consider a twice-continuously differentiable function $h:\R^n\times\R^p \to \R$ defined as:
\begin{equation}
\label{eqn:hcomp}
    h(\mb{x},\bs{\xi}) = h_0(\mb{x}) - \frac{1}{2\mu}(\bs{\xi}-\kdh)^\top(\bs{\xi}-\kdh),
\end{equation}
with $\mu\in\R_{>0}$. We note that instead of adding the quadratic error term as we did in \eqref{eqn:Vcomp}, we have subtracted it.
Let us define the set $\C\subset\R^n\times\R^p$ as the 0-superlevel set of the function $h$:
\begin{equation}
\label{eqn:Ccomp}
    \C = \{(\mb{x},\bs{\xi})\in\R^n\times\R^p ~|~ h(\mb{x},\bs{\xi})\geq 0\},
\end{equation}
noting that $\C\subset\C_0\times\R^p$.
This enables the following theorem:

\begin{theorem}
Let $\C_0$ be the 0-superlevel set of a twice-continuously differentiable function $h_0:\R^n\to\R$ with $\derp{h_0}{\mb{x}}(\mb{x}) \neq \mb{0}$ when $h_0(\mb{x}) = 0$. If there exists a twice-continuously differentiable function $\kdh:\R^n\to\R^p$ and a globally Lipschitz\footnote{We note this assumption permits linear extended class $\mathcal{K}$ functions, i.e, $\alpha_0(r) = kr$ for some $k\in\R_{>0}$, which are often used in practice} function $\alpha_0\in\mathcal{K}_{\infty}^e$ such that \eqref{eq:safety_high_level} holds, then there exists a locally Lipschitz continuous controller $\mb{k}:\R^n\times\R^p\to\R^p$ such that the function $h$ defined in \eqref{eqn:hcomp} is a Barrier Function for the closed-loop system \eqref{eqn:xdotcl}-\eqref{eqn:xidotcl} on the set $\C$ defined in \eqref{eqn:Ccomp}. Moreover, if $(\mb{x}_0,\bs{\xi}_0)\in\C$, then $\bs{\varphi}_{\mb{x}}(t)\in\C_0$ for all $t\in I((\mb{x}_0,\bs{\xi}_0))$.
\end{theorem}

\begin{proof}
We observe that:
\begin{equation}
\label{eqn:gradh}
    \begin{bmatrix} \derp{h}{\mb{x}}(\mb{x},\bs{\xi}) \\  \derp{h}{\bs{\xi}}(\mb{x},\bs{\xi}) \end{bmatrix} = \begin{bmatrix} \derp{h_0}{\mb{x}}(\mb{x}) +\frac{1}{\mu}(\bs{\xi}-\kdh)^\top\dkhdx\\ -\frac{1}{\mu}(\bs{\xi}-\kdh)^\top \end{bmatrix},
\end{equation}
from which we may conclude that if $\derp{h}{\bs{\xi}}(\mb{x},\bs{\xi}) = \mb{0}$ and $h(\mb{x},\bs{\xi})=0$, we must have $h_0(\mb{x}) = 0$, and thus $\derp{h}{\mb{x}}(\mb{x},\bs{\xi}) = \derp{h_0}{\mb{x}}(\mb{x}) \neq \mb{0}$ by assumption.
Furthermore, taking the time derivative of $h$ yields:
\begin{align}
% \label{eqn:hcomptimederiv}
    \dot{h}(\mb{x},\bs{\xi},\mb{u}) = ~&
    \dhdx \big( \dotx \big) \label{eqn:hcomp_dt}\\
    &- \frac{1}{\mu} (\bs{\xi} - \kdh)^\top \bigg( \dotxi \nonumber \\
    &\qquad\qquad - \dkhdx \big( \dotx \big) \bigg). \nonumber
\end{align}
Using a locally Lipschitz continuous feedback controller $\mb{k}:\R^n\times\R^p\to\R^m$ defined as:
\begin{align}
\label{eqn:kcbfcomp}
    \hspace{-0.4 mm} \mb{k}(\mb{x},\bs{\xi}) = &~
    \mb{g}_1(\mb{x},\bs{\xi})^\dagger
    \bigg( - \mb{f}_1(\mb{x},\bs{\xi})
    + \dkhdx \big( \dotx \big) \nonumber\\
    & + \mu \left(\dhdx \mb{g}_0(\mb{x}) \right)^\top
    - \frac{\lambda}{2} (\bs{\xi} - \kdh) \bigg),
\end{align}
with $\lambda\in\R_{\geq 0}$ yields:
\begin{align} \label{eqn:hdotalphas}
    \dot{h}(\mb{x},\bs{\xi},\mb{k}(\mb{x},\bs{\xi})) = ~& \dhdx(\mb{f}(\mb{x})+\mb{g}(\mb{x})\kdh) \\ &+\frac{\lambda}{2 \mu}
    (\bs{\xi} - \kdh)^\top (\bs{\xi} - \kdh), \nonumber \\
    \geq &~-\alpha_0(h_0(\mb{x})) +\frac{\lambda}{2 \mu}
    \|\bs{\xi} - \kdh\|^2_2.
\end{align}
Let $L$ be the Lipschitz constant of $\alpha_0$. Choosing $\lambda\geq L$, we have that:
\begin{align}
    \dot{h}(\mb{x},\bs{\xi},\mb{k}(\mb{x},\bs{\xi}))
    \geq &~-\alpha_0(h_0(\mb{x}))
    +\frac{L}{2 \mu} \|\bs{\xi} - \kdh\|^2_2,
    \label{eqn:proofsep}
    % \\ &~+\frac{L}{2 \mu}
    % (\bs{\xi} - \kd)^\top (\bs{\xi} - \kd), \nonumber
\end{align}
and the global Lipschitz property of $\alpha_0$ yields that:
\begin{align}
    \bigg\vert \alpha_0\bigg(h_0(\mb{x})-\frac{1}{2 \mu} \|\bs{\xi} - \kdh\|^2 &\bigg)-\alpha_0(h_0(\mb{x}))\bigg\vert \nonumber \\
    &\leq \frac{L}{2 \mu} \|\bs{\xi} - \kdh\|^2_2. \label{eqn:lipproof}
\end{align}
% \begin{align}
%     \bigg\vert \alpha_0\bigg(h_0(\mb{x})-&\frac{1}{2 \mu}
%     (\bs{\xi} - \kd)^\top (\bs{\xi} - \kd) \bigg)-\alpha_0(h_0(\mb{x}))\bigg\vert \nonumber \\
%     &\leq L\left\vert-\frac{1}{2 \mu}
%     (\bs{\xi} - \kd)^\top (\bs{\xi} - \kd)\right\vert, \nonumber \\
%     &= \frac{L}{2 \mu}
%     (\bs{\xi} - \kd)^\top (\bs{\xi} - \kd). \label{eqn:lipproof}
% \end{align}
\vspace{-0.3cm}
Noting the definition of \eqref{eqn:hcomp}, we may rearrange \eqref{eqn:lipproof} to yield:
\begin{equation}
\alpha_0(h(\mb{x},\bs{\xi})) \geq \alpha_0(h_0(\mb{x}))- \frac{L}{2 \mu} \| \bs{\xi} - \kdh \|^2_2.   
% \alpha_0(h(\mb{x},\bs{\xi})) \geq \alpha_0(h_0(\mb{x}))- \frac{L}{2 \mu} (\bs{\xi} - \kd)^\top (\bs{\xi} - \kd). \nonumber   
\end{equation}
Negating both sides of this expression and combining with \eqref{eqn:proofsep} allows us to conclude that:
\begin{equation}
\label{eqn:hdotprooffin}
    \dot{h}(\mb{x},\bs{\xi},\mb{k}(\mb{x},\bs{\xi})) \geq -\alpha_0(h(\mb{x},\bs{\xi})).
\end{equation}
Thus, $h$ is a BF for the closed-loop system \eqref{eqn:xdotcl}-\eqref{eqn:xidotcl} on the set $\C$. Hence, by Theorem \ref{thm:bfsafety} we may conclude the set $\C$ is safe, i.e.,
$(\mb{x}_0,\bs{\xi}_0)\in\C\implies \bs{\varphi}(t)\in \C \implies \bs{\varphi}_{\mb{x}}(t)\in\C_0$
% $(\mb{x}_0,\bs{\xi}_0)\in\C\implies \bs{\varphi}(t)\in \C$. Thus we have that $h(\bs{\varphi}(t))\geq 0$, and by construction we have that:
% \begin{equation}
%     h_0(\bs{\varphi}_\mb{x}(t)) \geq \frac{1}{2\mu} (\bs{\varphi}_{\bs{\xi}}(t) - \mb{k}_0(\bs{\varphi}_{\mb{x}}(t)))^\top (\bs{\varphi}_{\bs{\xi}}(t) - \mb{k}_0(\bs{\varphi}_{\mb{x}}(t))).
% \end{equation}
% As the right-hand side is non-negative, we have that $ h_0(\bs{\varphi}_\mb{x}(t))\geq 0$, or $\bs{\varphi}_{\mb{x}}(t)\in\C_0$
for all $t\in I((\mb{x}_0,\bs{\xi}_0))$.
\end{proof}

\begin{remark}
The preceding result establishes the safety of the set $\C$, rather than the set $\C_0$. We do not necessarily have that $\mb{x}_0\in\C_0$ implies $\bs{\varphi}_{\mb{x}}(t)\in\C_0$ for all $t\in I((\mb{x}_0,\bs{\xi}_0))$. The further requirement on the initial condition $\bs{\xi}_0$ is expected, and appears in other results studying safety for higher-order systems \cite{nguyen2016exponential, xiao2019control, xiao2021high}.
\end{remark}

We now make the following observation.
Suppose that $h_0(\mb{x}^*) = 0$, $\bs{\xi}=\mb{k}_0(\mb{x}^*)$ and:
\begin{equation}
\derp{h_0}{\mb{x}}(\mb{x}^*)(\mb{f}_0(\mb{x}^*)+\mb{g}_0(\mb{x}^*)\mb{k}_0(\mb{x}^*)) = 0,
\end{equation}
for some $\mb{x}^*\in\C_0$. Then, we have that:
% Suppose that for some $\mb{x}^*\in\C_0$ such that $h_0(\mb{x}^*) = 0$, we have that:
% \begin{equation}
%     \derp{h_0}{\mb{x}}(\mb{x}^*)\left(\mb{f}_0(\mb{x}^*)+\mb{g}_0(\mb{x}^*)\mb{k}_0(\mb{x}^*)\right) =  -\alpha_0(h_0(\mb{x}^*)) = 0. \nonumber
% \end{equation}
% If we have that $\bs{\xi}=\mb{k}_0(\mb{x}^*)$, such that $h(\mb{x}^*,\mb{k}_0(\mb{x}^*)) =0$, then from \eqref{eqn:hcomptimederiv} we have that:
% \begin{equation}
%     \dot{h}(\mb{x},\mb{k}_0(\mb{x}^*),\mb{u}) = -\alpha_0(h_0(\mb{x}^*)) =  -\alpha_0(h(\mb{x}^*,\mb{k}_0(\mb{x}^*))) = 0, \nonumber
% \end{equation}
% for all $\mb{u}\in\R^m$. In particular, we have that:
\begin{equation}
    \sup_{\mb{u}\in\R^m} \dot{h}(\mb{x}^*,\mb{k}_0(\mb{x}^*),\mb{u}) = 0 = -\alpha(h(\mb{x}^*,\mb{k}_0(\mb{x}^*))), 
\end{equation}
for any $\alpha\in\mathcal{K}_{\infty}^e$. Thus, we do not have that there exists an extended class $\mathcal{K}_\infty$ function $\alpha$ such that the strict inequality in \eqref{eq:cbf} is met, and hence we may not conclude that $h$ is a CBF for the system \eqref{eqn:xdot}-\eqref{eqn:xidot} on $\C$. The primary reason that $h$ is not a CBF lies in the fact that when $\bs{\xi}=\mb{k}_0(\mb{x}^*)$, the input does not have an effect on the time derivative of $h$. In this situation, the evolution of $h$ is entirely dependent on the design of the controller $\mb{k}_0$. Suppose that instead of \eqref{eq:safety_high_level}, we have that:
\begin{equation}
    \derp{h_0}{\mb{x}}(\mb{x})\left(\mb{f}_0(\mb{x})+\mb{g}_0(\mb{x})\mb{k}_0(\mb{x}) \right)> -\alpha_0(h_0(\mb{x})).
    \label{eq:safety_high_level_strict}
\end{equation}
Considering any $\mb{x}\in\R^n$ now, if $\bs{\xi}=\mb{k}_0(\mb{x})$, we have that:
\begin{multline}
    \dot{h}(\mb{x},\kdh,\mb{u}) > -\alpha_0(h_0(\mb{x})) =  -\alpha_0(h(\mb{x},\mb{k}_0(\mb{x}))), 
\end{multline}
for all $\mb{u}\in\R^m$. Noting that if $\bs{\xi}\neq\mb{k}_0(\mb{x})$, $\dot{h}$ can be made arbitrarily large through input, we may conclude that:
\begin{equation}
\label{eqn:hdotsup}
    \sup_{\mb{u}\in\R^m} \dot{h}(\mb{x},\bs{\xi},\mb{u}) > -\alpha_0(h(\mb{x},\bs{\xi})).
\end{equation}
This is summarized in the following theorem:
\begin{theorem}
\label{thm:cbfbackstep}
Let $\C_0$ be the 0-superlevel set of a twice-continuously differentiable function $h_0:\R^n\to\R$ with $\derp{h_0}{\mb{x}}(\mb{x}) \neq \mb{0}$ when $h_0(\mb{x}) = 0$. If there exists a twice-continuously differentiable function $\kdh:\R^n\to\R^p$ and a function $\alpha_0\in\mathcal{K}_{\infty}^e$ such that \eqref{eq:safety_high_level_strict} holds, then the function $h$ defined in \eqref{eqn:hcomp} is a Control Barrier Function for the system \eqref{eqn:xdot}-\eqref{eqn:xidot} on the set $\C$ defined in \eqref{eqn:Ccomp}.
\end{theorem}
Theorem \ref{thm:cbfbackstep} does not explicitly require the assumption of global Lipschitz continuity on $\alpha_0$, which was needed to achieve \eqref{eqn:hdotprooffin} when using the particular controller \eqref{eqn:kcbfcomp}. As CBFs are typically used in the context of control synthesis (beyond purely verification), we notice that \eqref{eqn:hdotsup} implies that:
\begin{equation}
    \sup_{\mb{u}\in\R^m}\dot{h}(\mb{x},\bs{\xi},\mb{u}) > -\alpha_1(h(\mb{x},\bs{\xi})),
\end{equation}
for any $\alpha_1\in\mathcal{K}_\infty^e$ such that $\alpha_1(s)\geq\alpha_0(s)$ for all $s\in\R$. Thus we may view $\alpha_1$ as an design parameter we may specify. For any such locally Lipschitz\footnote{Though it is not necessary for $\alpha_0$ to be locally Lipschitz continuous to imply the existence of such an $\alpha_1$, it is a sufficient condition.} $\alpha_1\in\mathcal{K}_\infty^e$ and any locally Lipschitz continuous $\mb{k}_{\rm d}:\R^n\times\R^p\to\R^m$, we can synthesize an optimization-based controller:
\begin{align}
\label{eqn:cbfbackstepqp}
    \mb{k}(\mb{x},\bs{\xi}) =  \argmin_{\mb{u}\in\R^m}&~ \frac{1}{2}\Vert \mb{u}-\mb{k}_{\rm d}(\mb{x},\bs{\xi}) \Vert_2^2 \\ \textrm{s.t.}&~ \dot{h}(\mb{x},\bs{\xi},\mb{u}) \geq -\alpha_1(h(\mb{x},\bs{\xi})), \nonumber
\end{align}
that is locally Lipschitz continuous on $\R^n\times\R^p$ \cite{jankovic2018robust} and renders $h$ a BF for \eqref{eqn:xdotcl}-\eqref{eqn:xidotcl} on $\C$.

\section{Multi-Step CBF Backstepping}
\label{sec:multibarbackstep}
In this section we extend the preceding CBF backstepping approach to higher-order mixed-relative degree systems via a recursive design process typical of backstepping.

Consider the nonlinear system\footnote{We do not notate a closed-loop system, but assume it is understood that when we refer to this system as closed-loop, it is operating under a controller.}  in strict feedback form:
\begin{subequations}\label{eq:cascade_many}
\begin{align}
 \dot{\bs{\xi}}_0 &= \mb{f}_0(\bs{\xi}_0)+\mb{g}_{0,\bs{\xi}}(\bs{\xi}_0)\bs{\xi}_1+\mb{g}_{0,\mb{u}}(\bs{\xi}_0) \mb{u}_0, \\
%  = \begin{bmatrix}\mb{k}_{01}(\mb{x}) \\ \mb{k}_{02}(\mb{x}) \end{bmatrix} = \mb{k}_0(\mb{x}),\\
\dot{\bs{\xi}}_1 &= \mb{f}_1(\bs{\xi}_0,\bs{\xi}_1)+\mb{g}_{1,\bs{\xi}}(\bs{\xi}_0,\bs{\xi}_1)\bs{\xi}_2+\mb{g}_{1,\mb{u}}(\bs{\xi}_0,\bs{\xi}_1) \mb{u}_1, \\
\vdots & \nonumber\\
\dot{\bs{ \xi}}_r &= \mb{f}_r(\bs{\xi}_0,\bs{\xi}_1,\bs{\xi}_2,\dots \bs{\xi}_r)+\mb{g}_{r}(\bs{\xi}_0,\bs{\xi}_1,\dots,\bs{\xi}_r)\mb{u}_{r},
\end{align}
\end{subequations}
with states $\bs{\xi}_i\in\R^{p_i}$ and inputs $\mb{u}_i\in\R^{m_i}$ for $i=0,\ldots,r$. The functions $\mb{f}_i$, $\mb{g}_{i,\mb{u}}$ for $i=0,\ldots,r$ and $\mb{g}_{i,\bs{\xi}}$ for $i=0,\ldots,r-1$ are assumed to be smooth on their respective domains. We further assume that the functions $\mb{g}_i=(\mb{g}_{i,\bs{\xi}},\mb{g}_{i,\mb{u}})$ for $i=1,\ldots,r-1$ and the function $\mb{g}_{r}$ are pseudo-invertible on their respective domains. Let us denote $q_i\triangleq \sum_{j=0}^{i} p_j$,
$M_r\triangleq\sum_{j=0}^{r} m_j$, and
$\mb{z}_i \triangleq (\bs{\xi}_0,\bs{\xi}_1,\dots,\bs{\xi}_i) \in \R^{q_i}$ for $i=0,\ldots,r$. We seek to construct a controller $\map{\mb{k}}{\R^{q_r}}{\R^{M_r}}$ such that setting $\mb{u}=(\mb{u}_0,\dots,\mb{u}_{r})=\mb{k}(\mb{z}_r)$ achieves safety.

Suppose the set $\C_0$ is defined as the 0-superlevel set of a smooth function $h_0:\R^{q_0}\to\R$ as in \eqref{eq:safeset}, with $\derp{h_0}{\bs{\xi}_0}(\mb{z}_0)\neq\mb{0}$ when $h_0(\mb{z}_0) = 0$. Let smooth functions $\mb{k}_{0,\bs{\xi}}:\R^{q_0}\to\R^{p_1}$ and $\mb{k}_{0,\mb{u}}:\R^{q_0}\to\R^{m_0}$, and a globally Lipschitz continuous function $\alpha_0\in\mathcal{K}_\infty^e$ with Lipschitz constant $L$ satisfy:
\begin{align}
\label{eqn:multisteptopsafe}
    \frac{\partial h_0}{\partial \bs{\xi}_0}(\mb{z}_0)\big(\mb{f}_0(\mb{z}_0)&+\mb{g}_{0,\bs{\xi}}(\mb{z}_0)\mb{k}_{0,\bs{\xi}}(\mb{z}_0)\\& +\mb{g}_{0,\mb{u}}(\mb{z}_0)\mb{k}_{0,\mb{u}}(\mb{z}_0)\big) \geq - \alpha_0(h_0(\mb{z}_0)), \nonumber
\end{align}
for all $\mb{z}_0\in\R^{q_0}$. Consider smooth functions (to be defined) $\mb{k}_{i,\bs{\xi}}:\R^{q_i}\to\R^{p_{i+1}}$ for $i=1,\ldots,r-1$ and $\mb{k}_{i,\mb{u}}:\R^{q_i}\to\R^{m_i}$ for $i=1,\ldots,r$, and define the smooth function $h:\R^{q_r}\to\R$:
\begin{equation}
\label{eqn:hrecursion}
    h(\mb{z}_{r}) = h_{0}(\mb{z}_{0})-\sum_{i=1}^{r}\frac{1}{2\mu_{i}}\Vert\bs{\xi}_{i}-\mb{k}_{i-1,\bs{\xi}}(\mb{z}_{i-1})\Vert_2^2,
\end{equation}
with $\mu_{i}\in\R_{>0}$ for $i=1,\ldots,r$. Define the set $\C\subset\R^{q_r}$ as:
\begin{equation}
\label{eqn:Cmanystep}
    \C = \{\mb{z}_r\in\R^{q_r} ~|~ h(\mb{z}_r) \geq 0\},
\end{equation}
noting that $\C\subseteq\C_0\times\R^{p_1}\times\cdots\times\R^{p_r}$. Given this construction, we have the following result:
\begin{theorem}
Let $\C_0$ be the 0-superlevel set of smooth function $h_0:\R^{p_0}\to\R$ with $\derp{h_0}{\bs{\xi}_0}(\mb{z}_0) \neq \mb{0}$ when $h_0(\mb{z}_0) = 0$. If there exist smooth functions $\mb{k}_{0,\bs{\xi}}:\R^{p_0}\to\R^{p_1}$ and $\mb{k}_{0,\mb{u}}:\R^{p_0}\to\R^{m_0}$ and a globally Lipschitz function $\alpha_0\in\mathcal{K}_{\infty}^e$ such that \eqref{eqn:multisteptopsafe} holds, then there exists a smooth controller $\mb{k}:\R^{q_r}\to\R^{M_r}$ and functions $\mb{k}_{i,\bs{\xi}}:\R^{q_i}\to\R^{p_{i+1}}$ for $i=1,\ldots,r-1$ such that the function $h:\R^{q_r}\to\R$ defined in \eqref{eqn:hrecursion} is a Barrier Function for the closed-loop system \eqref{eq:cascade_many} on the set $\C$ defined in \eqref{eqn:Cmanystep}. Moreover, if the initial condition $\mb{z}_{r,0}\in\C$, then $\bs{\varphi}_{\bs{\xi}_0}(t)\in\C_0$ for all $t\in I(\mb{z}_{r,0})$.
\end{theorem}

\begin{proof}
We observe that:
\begin{align}
    \derp{h}{\bs{\xi}_0}(\mb{z}_r) = & \derp{h_0}{\bs{\xi}_0}(\mb{z}_0)  \\ &
    +\sum_{j=1}^{r}\frac{1}{\mu_{j}}(\bs{\xi}_{j}-\mb{k}_{j-1,\bs{\xi}}(\mb{z}_{j-1}))^\top\derp{\mb{k}_{j-1,\bs{\xi}}}{\bs{\xi}_0}(\mb{z}_{j-1}), \nonumber
\end{align}
and for $i\in\{1,\ldots,r\}$, we have that:
\begin{align}
    \derp{h}{\bs{\xi}_i}(\mb{z}_r) =&~ -\frac{1}{\mu_i}(\bs{\xi}_i-\mb{k}_{i-1,\bs{\xi}}(\mb{z}_{i-1}))^\top \\
    &~+\sum_{j=i+1}^{r}\frac{1}{\mu_j}(\bs{\xi}_{j}-\mb{k}_{j-1,\bs{\xi}}(\mb{z}_{j-1}))^\top\derp{\mb{k}_{j-1,\bs{\xi}}}{\bs{\xi}_{i}}(\mb{z}_{j-1}). \nonumber
\end{align}
We can see recursively (backwards) that if $\derp{h}{\bs{\xi}_i}(\mb{z}_r) = \mb{0}$ for $i=1,\ldots,r$, then we must have $\bs{\xi}_i = \mb{k}_{i-1,\bs{\xi}}(\mb{z}_{i-1})$ for $i=1,\ldots,r$, and thus $h(\mb{z}_r) = h_0(\bs{\xi}_0)$ and $\derp{h}{\bs{\xi}_0}(\mb{z}_r) = \derp{h_0}{\bs{\xi}_0}(\mb{z}_0)$. As $\derp{h_0}{\bs{\xi}_0}(\mb{z}_0)\neq\mb{0}$ when $h_0(\mb{z}_0) = 0$, we have that $\derp{h}{\bs{\xi}_0}(\mb{z}_r)\neq\mb{0}$ when $h(\mb{z}_r) = 0$, such that $\derp{h}{\mb{z}_r}(\mb{z}_r)\neq\mb{0}$ when $h(\mb{z}_r) = 0$.

% Next we observe through lengthy computation that:
% \begin{align}
%     &\dot{h}(\mb{z}_r,\mb{u}_0,\ldots,\mb{u}_r) = \\& \frac{\partial h_0}{\partial \bs{\xi}_0}(\mb{z}_0)\big(\mb{f}_0(\mb{z}_0)+\mb{g}_{0,\bs{\xi}}(\mb{z}_0)\bs{\xi}_1 +\mb{g}_{0,\mb{u}}(\mb{z}_0)\mb{u}_{0}\big) \nonumber\\
%     &-\sum_{i=1}^{r-1}\frac{1}{\mu_{i}}(\bs{\xi}_{i}-\mb{k}_{i-1,\bs{\xi}}(\mb{z}_{i-1}))^\top\bigg(\mb{f}_{i}(\mb{z}_{i})+\mb{g}_{i,\bs{\xi}}(\mb{z}_{i})\bs{\xi}_{i+1} \nonumber \\ &  +\mb{g}_{i,\mb{u}}(\mb{z}_{i})\mb{u}_{i}-\sum_{j=0}^{i-1}\derp{\mb{k}_{i-1,\bs{\xi}}}{\bs{\xi}_j}(\mb{z}_{i-1})\big(\mb{f}_{j}(\mb{z}_{j})+\mb{g}_{j,\bs{\xi}}(\mb{z}_{j})\bs{\xi}_{j+1} \nonumber \\ &  +\mb{g}_{j,\mb{u}}(\mb{z}_{j})\mb{u}_{j} \bigg) - \frac{1}{\mu_{r}}(\bs{\xi}_{r}-\mb{k}_{r-1,\bs{\xi}}(\mb{z}_{r-1}))^\top\bigg(\mb{f}_{r}(\mb{z}_{i}) \nonumber \\ &  +\mb{g}_{r,\mb{u}}(\mb{z}_{r})\mb{u}_{r}-\sum_{j=0}^{r-1}\derp{\mb{k}_{r-1,\bs{\xi}}}{\bs{\xi}_j}(\mb{z}_{r-1})\big(\mb{f}_{j}(\mb{z}_{j}) \nonumber \\ &+\mb{g}_{j,\bs{\xi}}(\mb{z}_{j})\bs{\xi}_{j+1} +\mb{g}_{j,\mb{u}}(\mb{z}_{j})\mb{u}_{j} \big)
%     \bigg) \nonumber.
% \end{align}
Using $\mb{k}_{0,\bs{\xi}}$ and $\mb{k}_{0,\mb{u}}$, we define the smooth functions:
\begin{align}
    &\begin{bmatrix}\mb{k}_{1,\bs{\xi}}(\mb{z}_1) \\ \mb{k}_{1,\mb{u}}(\mb{z}_1) \end{bmatrix} = \mb{g}_1(\mb{z}_1)^\dagger\bigg(-\mb{f}_1(\mb{z}_1) + \mu_0\bigg(\derp{h_0}{\bs{\xi}_0}(\mb{z}_0)\mb{g}_{0,\bs{\xi}}(\mb{z}_0) \bigg)^\top \nonumber \\& +\derp{\mb{k}_{0,\bs{\xi}}}{\bs{\xi}_0}(\mb{z}_0)\big(\mb{f}_0(\mb{z}_0)+\mb{g}_{0,\bs{\xi}}(\mb{z}_0)\bs{\xi}_1 +\mb{g}_{0,\mb{u}}(\mb{z}_0)\mb{k}_{0,\mb{u}}(\mb{z}_0)\big) \nonumber \\ &-\frac{\lambda_1}{2}(\bs{\xi}_1-\mb{k}_{0,\bs{\xi}}(\mb{z}_0)) \bigg).
\end{align}
For $i=2,\ldots,r-1$, we recursively define the smooth functions:
\begin{align}
    \begin{bmatrix}\mb{k}_{i,\bs{\xi}}(\mb{z}_i) \\ \mb{k}_{i,\mb{u}}(\mb{z}_i) \end{bmatrix}& = \mb{g}_i(\mb{z}_i)^\dagger\bigg(-\mb{f}_i(\mb{z}_i)  \\ &-\mu_{i}\mb{g}_{i-1,\bs{\xi}}(\mb{z}_{i-1})^\top(\bs{\xi}_{i-1}-\mb{k}_{i-2,\bs{\xi}}(\mb{z}_{i-2})) \nonumber \\& +\sum_{j=0}^{i-1}\derp{\mb{k}_{i-1,\bs{\xi}}}{\bs{\xi}_j}(\mb{z}_{i-1})\big(\mb{f}_j(\mb{z}_j)+\mb{g}_{j,\bs{\xi}}(\mb{z}_j)\bs{\xi}_{j+1}\nonumber \\ & +\mb{g}_{j,\mb{u}}(\mb{z}_j)\mb{k}_{j,\mb{u}}(\mb{z}_j) \big) -\frac{\lambda_i}{2}(\bs{\xi}_i-\mb{k}_{i-1,\bs{\xi}}(\mb{z}_{i-1})) \bigg), \nonumber
\end{align}
and lastly define the smooth function:
\begin{align}
\mb{k}_{r,\mb{u}}(\mb{z}_i)& = \mb{g}_r(\mb{z}_r)^\dagger\bigg(-\mb{f}_r(\mb{z}_r)  \\ &-\mu_{r}\mb{g}_{r-1,\bs{\xi}}(\mb{z}_{r-1})^\top(\bs{\xi}_{r-1}-\mb{k}_{r-2,\bs{\xi}}(\mb{z}_{r-2})) \nonumber \\& +\sum_{j=0}^{r-1}\derp{\mb{k}_{r-1,\bs{\xi}}}{\bs{\xi}_j}(\mb{z}_{r-1})\big(\mb{f}_j(\mb{z}_j)+\mb{g}_{j,\bs{\xi}}(\mb{z}_j)\bs{\xi}_{j+1}\nonumber \\ & +\mb{g}_{j,\mb{u}}(\mb{z}_j)\mb{k}_{j,\mb{u}}(\mb{z}_j) \big) -\frac{\lambda_r}{2}(\bs{\xi}_r-\mb{k}_{r-1,\bs{\xi}}(\mb{z}_{r-1})) \bigg), \nonumber
\end{align}
Letting the controller $\mb{k}:\R^{q_r}\to\R^{M_r}$ be defined as:
\begin{equation}
    \mb{k}(\mb{z}_r) = \begin{bmatrix} \mb{k}_{0,\mb{u}}(\mb{z}_0)^\top & \cdots & \mb{k}_{r,\mb{u}}(\mb{z}_r)^\top\end{bmatrix}^\top,
\end{equation}
a sequence of (laborious) calculations yields:
\begin{equation}
    \dot{h}(\mb{z}_r,\mb{k}(\mb{z}_r)) \geq -\alpha_0(h_0(\mb{z}_0)) + \sum_{i=1}^{r}\frac{\lambda_i}{2\mu_{i}}\Vert\bs{\xi}_{i}-\mb{k}_{i-1,\bs{\xi}}(\mb{z}_{i-1})\Vert_2^2. \nonumber
\end{equation}
Choosing $\lambda_i\geq L$ for $i=1,\ldots,r$ and following the same argument as in \eqref{eqn:hdotalphas}-\eqref{eqn:hdotprooffin}, we arrive at:
% we have that:
% \begin{equation}
%     \dot{h}(\mb{z}_r,\mb{k}(\mb{z}_r)) \geq -\alpha_0(h_0(\mb{z}_0)) + \sum_{i=1}^{r}\frac{L}{2\mu_{i}}\Vert\bs{\xi}_{i}-\mb{k}_{i-1,\bs{\xi}}(\mb{z}_{i-1})\Vert_2^2. \nonumber
% \end{equation}
% Using the Lipschitz property of $\alpha_0$, we have that:
% \begin{equation}
%     \alpha_0(h(\mb{z}_r)) \geq \alpha_0(h_0(\mb{z}_0)) - \sum_{i=1}^{r}\frac{L}{2\mu_{i}}\Vert\bs{\xi}_{i}-\mb{k}_{i-1,\bs{\xi}}(\mb{z}_{i-1})\Vert_2^2.
% \end{equation}
% Negating both sides and using the definition of $h$ yields:
\begin{equation}
\dot{h}(\mb{z}_r,\mb{k}(\mb{z}_r)) \geq -\alpha_0(h(\mb{z}_r)).
\end{equation}
Thus, $h$ is a BF for the closed-loop system \eqref{eq:cascade_many} on the set $\C$. Hence, by Theorem \ref{thm:bfsafety} we may conclude the set $\C$ is safe, i.e.,
$\mb{z}_{r,0}\in\C\implies \bs{\varphi}(t)\in \C \implies \bs{\varphi}_{\bs{\xi}_0}(t)\in\C_0$. 
\end{proof}
If instead of \eqref{eqn:multisteptopsafe} we suppose that:
\begin{align}
\label{eqn:multisteptopstrictsafe}
    \frac{\partial h_0}{\partial \bs{\xi}_0}(\mb{z}_0)\big(\mb{f}_0(\mb{z}_0)&+\mb{g}_{0,\bs{\xi}}(\mb{z}_0)\mb{k}_{0,\bs{\xi}}(\mb{z}_0)\\& +\mb{g}_{0,\mb{u}}\mb{k}_{0,\mb{u}}(\mb{z}_0)\big) > - \alpha_0(h_0(\mb{z}_0)), \nonumber
\end{align}
we have the following result:
\begin{theorem}
\label{thm:cbfmultibackstep}
Let $\C_0$ be the 0-superlevel set of a smooth function $h_0:\R^{q_0}\to\R$ with $\derp{h_0}{\bs{\xi}_0}(\mb{z}_0) \neq \mb{0}$ when $h_0(\mb{z}_0) = 0$. If there exist smooth functions $\mb{k}_{0,\bs{\xi}}:\R^{p_0}\to\R^{p_1}$ and $\mb{k}_{0,\mb{u}}:\R^{p_0}\to\R^{m_0}$ and a globally Lipschitz continuous function $\alpha_0\in\mathcal{K}_{\infty}^e$ such that \eqref{eqn:multisteptopstrictsafe} holds, then the function $h$ defined in \eqref{eqn:hrecursion} is a Control Barrier Function for the system \eqref{eq:cascade_many} on the set $\C$ defined in \eqref{eqn:Cmanystep}.
\end{theorem}
Consequently, for any locally Lipschitz $\alpha_1\in\mathcal{K}_\infty^e$ such that $\alpha_1(s)\geq\alpha_0(s)$ for all $s\in\R$ and any locally Lipschitz continuous $\mb{k}_{\rm d}:\R^{q_r}\to\R^{M_r}$, we can synthesize a controller:
\begin{align}
\label{eqn:cbfbackstepqp_multi}
    \mb{k}(\mb{z}_r) =  \argmin_{\mb{u}\in\R^{M_r}}&~ \frac{1}{2}\Vert \mb{u}-\mb{k}_{\rm d}(\mb{z}_r) \Vert_2^2 \\ \textrm{s.t.}&~ \dot{h}(\mb{z}_r,\mb{u}_0,\ldots,\mb{u}_r) \geq -\alpha_1(h(\mb{z}_r)), \nonumber
\end{align}
that is locally Lipschitz continuous on $\R^{q_r}$ \cite{jankovic2018robust} and renders $h$ a BF for \eqref{eq:cascade_many} on $\C$.

\section{Joint CLF and CBF Backstepping}
\label{sec:design}
In this section we use joint Lyapunov and CBF backstepping to achieve both stability and safety of a cascaded system. 
For simplicity, let us consider the system \eqref{eqn:xdot}-\eqref{eqn:xidot}. Suppose there exists  functions $V_0:\R^n\to\R_{\geq 0}$, $h_0:\R^n\to\R$ and $\mb{k}_0:\R^n\to\R^p$ with $\mb{k}_0(\mb{0})=\mb{0}$, all twice-continuously differentiable, and functions $\gamma_1,\gamma_2,\gamma_3\in\mathcal{K}_{\infty}$ and a globally Lipschitz continuous function $\alpha_0\in\mathcal{K}_{\infty}^e$ such that \eqref{eqn:V0bds}-\eqref{eq:stability_high_level} and \eqref{eq:safety_high_level} are satisfied. Furthermore, let us define the set $\C_0\subset\R^n$ as in \eqref{eq:safeset}. As before, we wish to stabilize the state to the origin while ensuring it remains in the set $\C_0$. Let us construct twice-continuously differentiable functions $V:\R^n\times\R^p\to\R_{\geq 0}$ and $h:\R^n\times\R^p\to\R$ as:
\begin{align}
     V(\mb{x},\bs{\xi}) = & ~ V_0(\mb{x}) + \frac{1}{2 \mu_V} (\bs{\xi} - \kd)^\top (\bs{\xi} - \kd), \\  h(\mb{x},\bs{\xi}) = &~  h_0(\mb{x}) - \frac{1}{2 \mu_h} (\bs{\xi} - \kd)^\top (\bs{\xi} - \kd),
\end{align}
with $\mu_V,\mu_h\in\R_{>0}$. The time derivatives for $V$ and $h$ are given in \eqref{eqn:Vcomp_dt} and \eqref{eqn:hcomp_dt}, using their respective values $\mu_V$ and $\mu_h$. We express them compactly here as:
\begin{align}
    \dot{V}(\mb{x},\bs{\xi},\mb{u}) = & ~ b_{V,1}(\mb{x},\bs{\xi})+\frac{1}{\mu_V}\mb{a}_1(\mb{x},\bs{\xi})^\top\mb{u} \\
    \dot{h}(\mb{x},\bs{\xi},\mb{u})= & ~ b_{h,1}(\mb{x},\bs{\xi})-\frac{1}{\mu_h}\mb{a}_1(\mb{x},\bs{\xi})^\top\mb{u},
\end{align}
for functions $b_{V,1},b_{h,1}:\R^n\times \R^p\to\R$ and $\mb{a}_1:\R^n\times \R^p\to\R^m$. As we saw in the individual backstepping cases, it was possible to design (different) controllers such that the bounds on the derivatives in \eqref{eqn:dotVgammas} and \eqref{eqn:hdotalphas} were met. This implies that:
\begin{align}
    \inf_{\mb{u}\in\R^m} b_{V,1}(\mb{x},\bs{\xi})&+\frac{1}{\mu_V}\mb{a}_1(\mb{x},\bs{\xi})^\top\mb{u} \\ &\leq -\gamma_3(\Vert\mb{x}\Vert)-\gamma_3'(\Vert
    \bs{\xi} - \kd)\Vert_2), \nonumber\\
   \inf_{\mb{u}\in\R^m} -b_{h,1}(\mb{x},\bs{\xi})&+\frac{1}{\mu_h}\mb{a}_1(\mb{x},\bs{\xi})^\top\mb{u} \\ &\leq \alpha_0(h_0(\mb{x})) -\frac{\lambda}{2 \mu_h}
    \|\bs{\xi} - \kd\|^2_2, \nonumber
\end{align}
We can rewrite these two inequality constraints as:
\begin{align}
    \mb{a}_1(\mb{x},\bs{\xi})^\top\mb{u} \leq&~ c_{V,1}(\mb{x},\bs{\xi}), \label{eqn:jointstabconst}\\
    \mb{a}_1(\mb{x},\bs{\xi})^\top\mb{u} \leq&~ c_{h,1}(\mb{x},\bs{\xi}), \label{eqn:jointsafeconst}
\end{align}
for functions $c_{V,1},c_{h,1}:\R^n\times\R^p\to\R$. A key observation is that these constraints are mutually satisfiable, i.e, if we design a controller $\mb{k}:\R^n\times\R^p\to\R^m$ such that:
\begin{equation}
    \mb{a}_1(\mb{x},\bs{\xi})^\top\mb{k}(\mb{x},\bs{\xi}) \leq \min\{c_{V,1}(\mb{x},\bs{\xi}),c_{h,1}(\mb{x},\bs{\xi})\},
\end{equation}
for all $(\mb{x},\bs{\xi})\in\R^n\times\R^p$, then both \eqref{eqn:jointstabconst} and \eqref{eqn:jointsafeconst} are met. Thus under this controller, $V$ is a Lyapunov function and $h$ is a Barrier Function on $\C$ for the closed-loop system \eqref{eqn:xdotcl}-\eqref{eqn:xidotcl}, such that we may conclude both stability and safety. An optimization-based controller achieving this is defined as:
\begin{align}
    \mb{k}(\mb{x},\bs{\xi})  =  \argmin_{\mb{u}\in\R^m}&~ \frac{1}{2}\Vert \mb{u} \Vert_2^2 \\ \textrm{s.t.}&~  \mb{a}_1(\mb{x},\bs{\xi})^\top\mb{u} \leq \min\{c_{V,1}(\mb{x},\bs{\xi}),c_{h,1}(\mb{x},\bs{\xi})\}. \nonumber
\end{align}
The intuition behind the joint feasibility of these constraints is that the controller $\mb{k}_0$ has been designed to provide both stability and safety, and we are using the input $\mb{u}$ to drive $\bs{\xi}$ to $\mb{k}_0(\mb{x})$, thus benefiting both stability and safety. The challenge is then to design a continuously differentiable controller $\mb{k}_0$ satisfying both \eqref{eq:stability_high_level} and \eqref{eq:safety_high_level}. To accomplish this, we will use the techniques presented in \cite{ong2019universal}. We note that designing smooth stabilizing controllers via Lyapunov functions often faces challenges at the origin \cite{sontag1989universal}. With a cascaded system, we may encounter the origin of the top-level state without the entire state being at the origin. Thus, in this work we slightly relax \eqref{eq:stability_high_level} to ensure smoothness, in which case we achieve practical stability as opposed to asymptotic stability.

Suppose that we are given a smooth desired controller $\mb{k}_{0,{\rm d}}:\R^n\to\R^p$ we wish to implement at the top-level, that is not necessarily stable nor safe. Consider the top-level constraints:
\begin{align*}
     \derp{V_0}{\mb{x}}(\mb{x})(\mb{f}_0(\mb{x})+\mb{g}_0(\mb{x})(\mb{k}_{0,{\rm d}}(\mb{x})+\mb{v})) \leq &~ -\gamma_3(\Vert\mb{x}\Vert_2) \\ &~ +\delta\psi(\|\mb{x}\|_2), \\
     \derp{h_0}{\mb{x}}(\mb{x})\left(\mb{f}_0(\mb{x})+\mb{g}_0(\mb{x})(\mb{k}_{0,{\rm d}}(\mb{x})+\mb{v}) \right)\geq &~ -\alpha_0(h_0(\mb{x})).
\end{align*}
with $\delta\in\R_{> 0}$ and $\psi:\R\to\R_{\geq 0}$ a bump function defined as:
\begin{equation}
\label{eqn:bump}
\psi(s) = \begin{cases}\exp \left( -\frac{1}{\epsilon^2-s^2} \right), & s \in (-\epsilon,\epsilon),\\ 0, & {\rm otherwise}, \end{cases}
\end{equation}
with $\epsilon\in\R_{>0}$. We can rewrite these constraints as:
\begin{align}
    \mb{a}_{V,0}(\mb{x})^\top\mb{v}+b_{V,0}(\mb{x}) \leq&~ 0, \label{eqn:V0simp} \\
    \mb{a}_{h,0}(\mb{x})^\top\mb{v}+b_{h,0}(\mb{x}) \leq&~ 0, \label{eqn:h0simp}
\end{align}
for functions $\mb{a}_{V,0},\mb{a}_{h,0}:\R^n\to\R^p$ and $b_{V,0},b_{h,0}:\R^n\to\R$. Assuming $V_0$ is a CLF and $h_0$ is a CBF on $\C_0$ for \eqref{eqn:xdot} implies the set-valued functions $\Uc_V,\Uc_h:\R^n\to\mathcal{P}(\R^p)$ defined as:
\begin{equation}
\Uc_i(\mb{x}) = \setdefb{\mb{v}\in\R^p}{\mb{a}_{i,0}(\mb{x})^\top \mb{v} + b_{i,0}(\mb{x}) \leq 0},
\end{equation}
with $i \in \{V,h\}$ satisfy $\mathcal{U}_i(\mb{x})\neq\{\emptyset\}$ for all $\mb{x}\in\R^n$. Moreover, for simplicity let us assume that $\Uc_V(\mb{x})\cap \Uc_h(\mb{x}) \neq\{\emptyset\}$ for all $\mb{x}\in\R^n$, such that there exists a $\mb{v}$ that satisfies both \eqref{eqn:V0simp} and \eqref{eqn:h0simp} simultaneously. We note that if this is not possible, this construction can be done relaxing stability and enforcing safety as is common with combined CLF-CBF methods \cite{ames2019control}.

For a set $\mathcal{U}\subseteq\R^p$, define the Gaussian weighted centroid function $\map{\bs{\mu}}{\R^n}{\R^p}$ as:
\begin{equation}
\bs{\mu}(\mb{x};\Uc) \defeq \frac{\int_{ \Uc} \mb{v}\phi(\mb{x},\mb{v})d\mb{v}}{\int_{\Uc} \phi(\mb{x},\mb{v})d\mb{v\\}},
\end{equation}
where $\phi:\R^n\times\R^p\to\R_{\geq 0}$ is defined as:
\begin{align}
    \phi(\mb{x},\mb{v}) & = \frac{1}{\sqrt{2\pi}} {\rm e}^{-\|\mb{v}\|_2^2/(2\sigma(\mb{x}))},
\end{align}
with a smooth function $\sigma:\R^n\to\R_{\geq 0}$. As in \cite{ong2019universal}, we may synthesize a controller:
\begin{multline}\label{eq:smooth}
\mb{k}_0(\mb{x}) = \mb{k}_{0,{\rm d}}(\mb{x}) + \zeta(\rho(\mb{x}))(\bs{\mu}(\mb{x};\Uc_V)+\bs{\mu}(\mb{x};\Uc_h))\\+(1-\zeta(\rho(\mb{x})))\bs{\mu}(\mb{x};\Uc_V\cap\Uc_h),
\end{multline}
where $\map{\zeta}{\R}{[0,1]}$ is a smooth partition of unity function with $\zeta(s)=0$ for $s\leq0$ and $\zeta(s)=1$ for $s\geq 1$, and:
\begin{equation}
\rho(\mb{x}) = \frac{\mb{a}_{V,0}(\mb{x})^\top\mb{a}_{h,0}(\mb{x})}{\|\mb{a}_{V,0}(\mb{x})\|_2\|\mb{a}_{h,0}(\mb{x})\|_2},
\end{equation}
encodes the angle between $\mb{a}_{V,0}$ and $\mb{a}_{h,0}$. The Gaussian weighted centroid functions in~\eqref{eq:smooth} have closed-form solutions \cite{tallis1961moment,tallis1965plane}. The controller in~\eqref{eq:smooth} respects both constraints, i.e., $(\mb{k}_0(\mb{x}) - \mb{k}_{0,{\rm d}}(\mb{x}))\in\Uc_V(\mb{x})\cap \Uc_h(\mb{x})$. In addition, $\mb{k}_0$ is smooth if the functions $\mb{a}_{V,0}, \mb{a}_{h,0}, b_{V,0}$ and $b_{h,0}$ are smooth.

\section{Simulation}
\label{sec:sim}
We now demonstrate CBF backstepping with two examples.

\begin{example} \label{exmpl:double_integrator}
Consider the planar double integrator system:
\begin{align}
\begin{split}
    \dot{\mb{x}} = \bs{\xi}, \quad
    \dot{\bs{\xi}} = \mb{u},
\end{split}
\end{align}
with ${\mb{x}, \bs{\xi}, \mb{u} \in \R^2}$.
We intend to control the system to a goal position ${\mb{x}_{\rm g} \in \R^2}$ (such that ${\lim_{t \to \infty} \bs{\varphi}_{\mb{x}}(t) = \mb{x}_{\rm g}}$) while avoiding an obstacle centered at ${\mb{x}_{\rm O} \in \R^2}$ with radius ${R_{\rm O} \in \R_{>0}}$.
Collision-free behavior is captured by the safe set $\C_0$ with:
\begin{equation}
    h_0(\mb{x}) = \frac{1}{2} \left( \| \mb{x} - \mb{x}_{\rm O} \|_2^2 - R_{\rm O}^2 \right),
    \label{eq:example_CBF}
\end{equation}
that satisfies ${h_0(\mb{x}) = 0 \implies \dhdx = (\mb{x} - \mb{x}_{\rm O})^\top \neq \mb{0}}$.
To reach the goal $\mb{x}_{\rm g}$, we rely on the desired smooth controller $\mb{k}_{0,{\rm d}}(\mb{x}) = - K_{\rm p}(\mb{x} - \mb{x}_{\rm g})$ which is used to define $\mb{k}_{0}$ through the smooth safety filter in \eqref{eq:smooth}. This is used to define $h$ as in \eqref{eqn:hcomp}, which used with the desired controller $\mb{k}_{\rm d}(\mb{x},\bs{\xi}) = - K_{\rm v}(\bs{\xi} - \kd)$
in the quadratic-program safety filter \eqref{eqn:cbfbackstepqp}. 
% at the high level:
% \begin{equation}
%     \mb{k}_{0,{\rm d}}(\mb{x}) = - K_{\rm p}(\mb{x} - \mb{x}_{\rm g}),
%     \label{eq:example_nominal_controller}
% \end{equation}
% and that at the low level:
% \begin{equation}
%     \mb{k}_{\rm d}(\mb{x},\bs{\xi}) = - K_{\rm v}(\bs{\xi} - \kd).
% \end{equation}

% [We can specify these in detail once the smoothing section is ready.]

% Whereas safety is provided by backstepping with CBFs, which relies on the smooth high-level controller:
% \begin{equation}
%     \kd = \kh + F \bigg( -\frac{A(\mb{x})}{\|\mb{B}(\mb{x})\|} \bigg) \frac{\mb{B}(\mb{x})^\top}{\|\mb{B}(\mb{x})\|},
% \end{equation}
% and the low-level control law:
% \begin{equation}
%     \mb{k}(\mb{x},\bs{\xi}) = \kl - \max \bigg\{ 0, \frac{a_h(\mb{x},\bs{\xi})}{\|\mb{b}(\mb{x},\bs{\xi})\|} \bigg\} \frac{\mb{b}(\mb{x},\bs{\xi})^\top}{\|\mb{b}(\mb{x},\bs{\xi})\|}.
% \end{equation}

\begin{figure}
\centering
\includegraphics[scale=1.0]{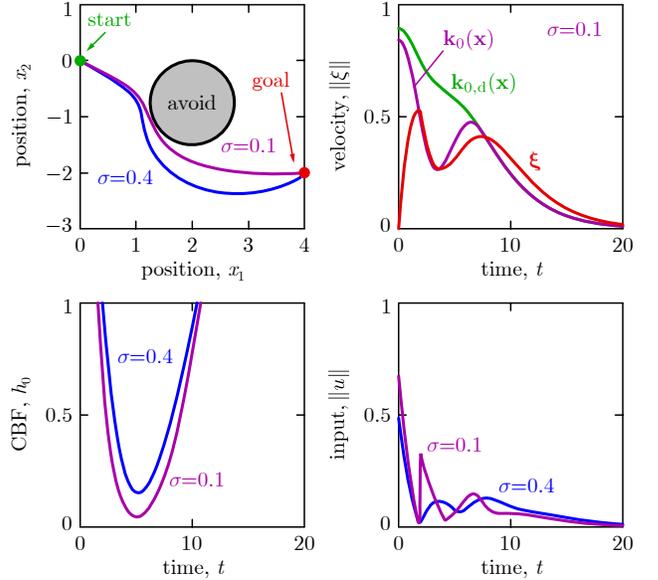}
\caption{
Obstacle avoidance with double integrator model via backstepping.
The system successfully avoids the obstacle and reaches the goal, while the conservatism of the route can be tuned by the smoothing parameter.
}
\label{fig:double_integrator}
\end{figure}

The closed-loop system is simulated in Fig.~\ref{fig:double_integrator} for
${K_{\rm p} = 0.2}$,
${K_{\rm v} = 0.8}$,
${\mu=1}$,
${\alpha_0(s) = \alpha_1(s) = s}$,
${\sigma \equiv 0.1}$ (purple) and
${\sigma \equiv 0.4}$ (blue).
The system safely reaches the goal without colliding with the obstacle.
% ($h_0$ is positive for all time).
As the smoothing parameter $\sigma$ is increased, the system takes a more conservative route farther from the obstacle.
This reduces the peak in the control input.
% caused by the safety filter.
\end{example}

\begin{example} \label{exmpl:unicycle}
Consider the planar unicycle model:
\begin{align}
\begin{split}
    \dot{x} = v \cos \psi, \quad
    \dot{y} = v \sin \psi, \quad
    \dot{\psi} = \omega.
\end{split}
\end{align}
where ${x,y,\psi,v,\omega \in \R}$.
This system can be written as:
\begin{align}
    \dot{\mb{x}} = \bs{\xi} u_0\triangleq\mb{w}, \quad
    \dot{\bs{\xi}} = \begin{bmatrix}
    -\xi_2 & \xi_1
    \end{bmatrix}^\top u_1,
    \label{eq:unicycle}
\end{align}
with
$\mb{x}=\begin{bmatrix}
x & y
\end{bmatrix}^\top$ and
$\bs{\xi}=\begin{bmatrix}
\cos \psi & \sin \psi
\end{bmatrix}^\top$.
% and
% $\mb{g}_1(\mb{x},\bs{\xi}) =
% \begin{bmatrix}
% -\xi_2 & \xi_1
% \end{bmatrix}^\top$.
% \begin{align}
%     \dot{\mb{x}} & = \mb{f}_0(\mb{x}) + \mb{g}_0(\mb{x}) \bs{\xi} u_0,
%     \label{eq:unicycle_high_level} \\
%     \dot{\bs{\xi}} & = \mb{f}_1(\mb{x},\bs{\xi})+\mb{g}_1(\mb{x},\bs{\xi}) u_1,
% \end{align}
% with
% $\mb{x}=\begin{bmatrix}
% x & y
% \end{bmatrix}^\top$,
% $\bs{\xi}=\begin{bmatrix}
% \cos \psi & \sin \psi
% \end{bmatrix}^\top$,
% $\mb{f}_0(\mb{x}) = \mb{0}$,
% $\mb{g}_0(\mb{x}) = \mb{I}$,
% $\mb{f}_1(\mb{x},\bs{\xi}) = \mb{0}$
% and
% $\mb{g}_1(\mb{x},\bs{\xi}) =
% \begin{bmatrix}
% -\xi_2 & \xi_1
% \end{bmatrix}^\top$.
% and:
% \begin{align}
% \begin{split}
%     % \mb{x} & =
%     % \begin{bmatrix}
%     % x \\ y
%     % \end{bmatrix}, \quad
%     \mb{f}_0(\mb{x}) & = \mb{0}, \quad
%     \mb{g}_0(\mb{x}) = \mb{I}, \\
%     % \bs{\xi} & =
%     % \begin{bmatrix}
%     % \cos \psi \\ \sin \psi
%     % \end{bmatrix}, \quad
%     \mb{f}_1(\mb{x},\bs{\xi}) & = \mb{0}, \quad
%     \mb{g}_1(\mb{x},\bs{\xi}) =
%     \begin{bmatrix}
%     -\xi_2 & \xi_1
%     \end{bmatrix}^\top.
% \end{split}
% \end{align}
Our goal is obstacle avoidance like in Example~\ref{exmpl:double_integrator}, via the CBF \eqref{eq:example_CBF}.

The unicycle model is in the form of \eqref{eq:cascade_many} except for an additional nonlinearity: the product of the heading direction $\bs{\xi}$ and the speed $u_0$ that gives the velocity vector ${\mb{w} = \bs{\xi} u_0}$.
With some care, this nonlinearity can be handled as follows.
First, notice that \eqref{eq:unicycle} is affine in both $\mb{w}$ and $u_0$.
Thus, a safe value $\mb{k}_0(\mb{x})$ for the velocity $\mb{w}$ can be designed such that it satisfies \eqref{eq:safety_high_level},
which is the same as $\mb{k}_0(\mb{x})$ in Example~\ref{exmpl:double_integrator}.
% We achieve this with the smooth safety filter of Example~\ref{exmpl:double_integrator} and the desired input $\mb{k}_{0,{\rm d}}(\mb{x}) = - K_{\rm p}(\mb{x} - \mb{x}_{\rm g})$.
We convert the safe velocity $\mb{k}_0(\mb{x})$ into a safe heading direction ${\mb{k}_{0,\bs{\xi}}(\mb{x}) = \mb{k}_0(\mb{x})/\|\mb{k}_0(\mb{x})\|_2}$ and safe speed ${k_{0,u}(\mb{x}) = \|\mb{k}_0(\mb{x})\|_2}$
% If there exists ${\bs{\kappa}^{-1}: \R^2 \to \R^3}$ such that ${\bs{\kappa} \big( \bs{\kappa}^{-1}(\mb{w}) \big) = \mb{w}}$, then one may use ${\begin{bmatrix} \mb{k}_{0,\bs{\xi}}(\mb{x})^\top & k_{0,u}(\mb{x}) \end{bmatrix}^\top = \bs{\kappa}^{-1}(\mb{k}_0(\mb{x}))}$.
by restricting to ${\mb{k}_0(\mb{x}) \neq \mb{0}}$.
% ${{\kappa}^{-1}(\mb{w}) = \begin{bmatrix} \mb{w}^\top/\|\mb{w}\| & \| \mb{w} \|\end{bmatrix}^\top}$.
% This results in a safe speed $k_{0,u}(\mb{x})$ and a safe heading direction $\mb{k}_{0,\bs{\xi}}(\mb{x})$. The latter
Then, $\mb{k}_{0,\bs{\xi}}(\mb{x})$ is incorporated into the composite barrier function $h$ in \eqref{eqn:hrecursion}.
By denoting the safe heading angle as $\psi_0(\mb{x})$, i.e., by writing
${\mb{k}_{0,\bs{\xi}}(\mb{x}) = \begin{bmatrix}
\cos \psi_0(\mb{x}) & \sin \psi_0(\mb{x})
\end{bmatrix}^\top}$,
we get:
\begin{equation}
    h(\mb{x},\bs{\xi}) = h_0(\mb{x}) - \frac{1}{\mu} \big( 1 - \cos(\psi - \psi_0(\mb{x})) \big),
\end{equation}
that gives penalty to heading in unsafe directions.
Then, we synthesize the controller ${\mb{u} = \begin{bmatrix}
u_0 & u_1 \end{bmatrix}^\top = \mb{k}(\mb{x},\bs{\xi})}$ via backstepping based on \eqref{eqn:cbfbackstepqp_multi}, where we use the desired controller
$\mb{k}_{\rm d}(\mb{x},\bs{\xi}) =
\begin{bmatrix}
K_{\rm p} \|\mb{x} - \mb{x}_{\rm g}\|_2 &
- K_\psi \big( \sin \psi - \sin \psi_0(\mb{x}) \big) \end{bmatrix}^\top$.
% \begin{equation}
%     \mb{k}_{\rm d}(\mb{x},\bs{\xi}) =
%     % \begin{bmatrix}
%     % K_{\rm p} \|\mb{x} - \mb{x}_{\rm g}\| \\
%     % - K_\psi (\xi_2 - k_{0,\bs{\xi},2}(\mb{x}))
%     % \end{bmatrix} =
%     \begin{bmatrix}
%     K_{\rm p} \|\mb{x} - \mb{x}_{\rm g}\| \\
%     - K_\psi \big( \sin \psi - \sin \psi_0(\mb{x}) \big)
%     \end{bmatrix}.
% \end{equation}

The behavior of the closed-loop system is shown by simulation results in Fig.~\ref{fig:unicycle} for
${K_{\rm p} = 0.2}$,
${K_{\rm \psi} = 3}$,
${\mu=1}$,
${\alpha_0(s) = \alpha_1(s) = s}$,
${\sigma \equiv 0.1}$ (purple) and
${\sigma \equiv 0.4}$ (blue).
Again, safety is guaranteed and more conservative smoothing makes the unicycle take a longer route.
% around the obstacle.
We remark that safety could also be enforced without backstepping, by relying on the input $u_0$ (speed) only.
Then, the input $u_1$ (angular velocity) would not be constrained and could be chosen freely.
This would result in the unicycle stopping in front of the obstacle and not reaching the goal (see black trajectory).
As opposed, backstepping synthesizes a barrier function $h$ such that inputs at all levels are utilized for safety.
Such barrier synthesis is nontrivial, and backstepping provides a systematic solution.

\begin{figure}
\centering
\includegraphics[scale=1.0]{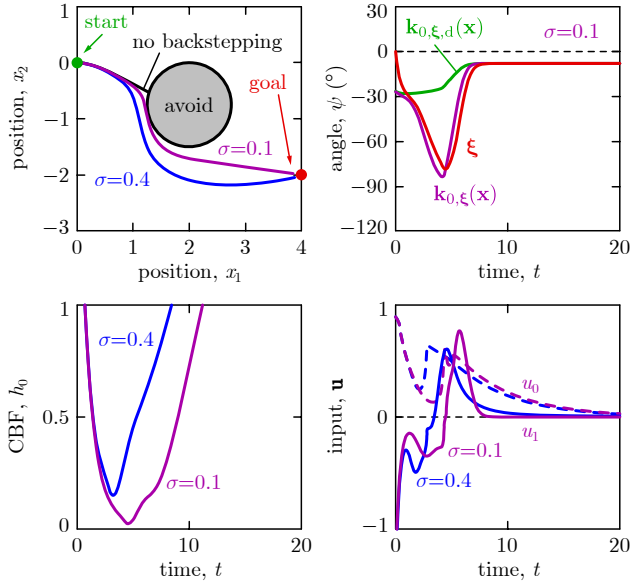}
\caption{
Obstacle avoidance with unicycle model via backstepping.
Remarkably, the unicycle is able to drive around the obstacle, while a standard safety filter without backstepping makes the unicycle stop in front of the obstacle.
}
\label{fig:unicycle}
\end{figure}

\end{example}

\section{Conclusion}
In conclusion, we have proposed a novel approach for using backstepping with Control Barrier Functions to design safety-critical controllers for nonlinear systems. Moreover, we unified this approach with Control Lyapunov Functions to achieve both stability and safety. Future work includes considering alternative methods for the smooth design of top-level controllers that are stabilizing and safe, and exploring the robustness to parameter uncertainty seen with backstepping.

\bibliographystyle{IEEEtran} 
\bibliography{main}

\end{document}